\documentclass{article}
\usepackage{graphicx}
\usepackage{color,framed}
\usepackage{amsmath}
\usepackage{amssymb}
\usepackage{amsfonts}
\usepackage{amsopn,amsthm}
\usepackage{bbm}
\usepackage{hyperref}
\usepackage[margin=1in]{geometry}

\newtheorem{theorem}{\bf Theorem}

\newtheorem{corollary}{\bf Corollary}
\newtheorem{definition}{\bf Definition}

\newtheorem{remark}{\bf Remark}

\def\mc{{\mathcal C}}

\def\e{\mathbb{E}}

\def\th{\mathbf{\mathtt g}}

\newcommand{\supp}{\text{supp}}

\newcommand{\E}{\mathbb{E}}

\newcommand{\ket}[1]{|#1\rangle}
\newcommand{\bra}[1]{\langle#1|}
\def\Pr{\text{\rm{Pr}}}

\def\tx{{\mathbf{\mathtt x}}}

\newcommand{\tr}{\text{\rm{tr}}}
\allowdisplaybreaks
\begin{document}

\title{Quantum Achievability Proof via Collision Relative Entrop}
\author{Salman Beigi$^1$ and Amin Gohari$^{1,2}$
\\ \emph{\small $^1$School of Mathematics, Institute for Research in Fundamental Sciences (IPM), Tehran, Iran}\\
\emph{\small $^2$Department of Electrical Engineering, Sharif University of Technology, Tehran, Iran}}

\date{January 11, 2017}

\maketitle

\begin{abstract}
In this paper, we provide a simple framework for deriving one-shot achievable bounds for some problems in quantum information theory. Our framework is based on the joint convexity of the exponential of the collision relative entropy, and is a (partial) quantum generalization of the technique of Yassaee et al.\ (2013) from classical information theory. Based on this framework, we derive one-shot achievable bounds for the problems of communication over classical-quantum channels, quantum hypothesis testing, and classical data compression with quantum side information. We argue that our one-shot achievable bounds are strong enough to give the asymptotic achievable rates of these problems even up to the second order.  
 \end{abstract}

\section{Introduction} 
Yassaee et al.~\cite{Yassaee} have recently proposed a general technique for proving upper bounds on the probability of error 
for classical network information theory problems in the \emph{one-shot} case. By one-shot we mean a setup whose goal is to find a strategy for (say) transmitting one out of $M$ messages in a \emph{single use} of the network with (average or maximal) probability of error being as small as possible. This is unlike the traditional setup where the focus is on asymptotic rates and vanishing error probability.
Some salient features of the method of \cite{Yassaee} are as follows: (1) it provides one-shot results whose form resembles (in a systematic way) that of asymptotic results. In particular, the denominator of the one-shot bounds given by this technique consists of a product of multiple terms, each of which resembles the covering or packing lemmas of the asymptotic regime; (2) it fundamentally differs from the traditional methods where first various error events are identified, and then union bound (and packing or covering lemmas on individual error events) is used. Instead, the technique uses \emph{Jensen's inequality} only and is able to consider the effect of all the error events at once (hence yielding stronger results). Further, since one-shot versions of mutual packing and covering lemma are not known to exist, this technique outperforms the traditional ones;
(3) the technique is able to recover the second-order asymptotics (dispersion) of the point-to-point channel capacity (i.e., a channel with one transmitter and one receiver) and provides new finite blocklength achievability results for many other network problems.

The technique of~\cite{Yassaee} can roughly be explained as follows. For the problem of point-to-point channel capacity, for example, this technique (as usual) uses a random coding to encode the message.  However, to decode the message at the receiver's side, it does not use the maximum likelihood (ML) decision rule. Instead, the technique uses a decoder called \emph{stochastic likelihood coder} (SLC) in~\cite{Yassaee}. To describe the difference between an SLC and an ML decoder, assume that to send a message $m\in \{1,\dots, M\}$ over a channel with input $X$ and output $Y$ described by conditional distribution $P_{Y|X}$, we use codewords of length one. 
So each message $m$ is mapped to an input symbol $x_m$ at the encoder. Then having received $y$, to encode the message an ML decoder picks $\hat{m}$ such that $P_{M|Y}(\hat{m}|y)$ is maximized. Here it is assumed that the message $m\in \{1,\dots, M\}$ is chosen uniformly at random, so we have a joint distribution $P_{MXY}$ and the conditional distribution $P_{M|Y}$ is computed with respect to this joint distribution.

On the other hand, to decode the message from output $Y=y$, an SLC picks $\hat m$ randomly according to the conditional distribution $P_{M|Y=y}$. That is, an SLC decodes $Y=y$ to $\hat m$ with probability
\begin{align}
P_{M|Y}(\hat{m}|y)&=\dfrac{P_{M,Y}(\hat{m},y)}{P_{Y}(y)}=\dfrac{P_{Y|X}(y|x_{\hat{m}})}{\sum_{m'}P_{Y|X}(y|x_{m'})},\label{eq:SLCeq}
\end{align}
where for the equality we use the fact that the distribution on $\{1,\dots, M\}$ is uniform which induces the uniform distribution on the codebook $\{x_1,\dots, x_M\}$.
This decoder is called stochastic likelihood coder in~\cite{Yassaee} as it is a probabilistic decoder and uses the likelihood of occurrence of messages to choose them.  

Putting in the terminology of quantum information theory, SLC is nothing but a pretty good measurement (also called square-root measurement), or a transpose channel in general.  To see this, consider a classical-quantum channel that under input $x$ outputs quantum state 
$\rho_{x}$. Then for a codebook $\{x_1,\dots, x_M\}$, pretty good measurement is a measurement with POVM elements
\begin{align*}
E_{\hat m}&=\Big(\sum_{m'}\rho_{x_{m'}}\Big)^{-\frac{1}{2}}\rho_{x_{\hat m}}\Big(\sum_{m'}\rho_{x_{m'}}\Big)^{-\frac{1}{2}}.
\end{align*}
Now assuming that $\rho_x$ is a diagonal density matrix with diagonal elements $P_{Y|X}(y|x)$, i.e., the channel is classic, then $E_{\hat m}$ is a diagonal matrix with diagonal elements given by~\eqref{eq:SLCeq}. So SLC is indeed a special case of pretty good measurements in quantum theory. For the definition of transpose channel and to verify that it is a generalization of SLC see e.g.,~\cite{NgM}.

In the technique proposed by \cite{Yassaee}, the expected value of the probability of successful decoding for a random codebook  is bounded from below using Jensen's inequality. This simple technique is general enough to be applied to several problems in classical network information theory, and interestingly gives a tight asymptotic bound up to the second order for the point-to-point channel capacity problem.\\

\noindent \textbf{Our contributions:}
Motivated by the appealing features of the technique of \cite{Yassaee}, we are interested to see if it extends to quantum information theory. In this paper we provide a partial generalization that in particular can be applied to the problems of communication over classical-quantum (c-q) channels, quantum hypothesis testing, and classical data compression with quantum side information. To analyze these  problems, we observe that the probability of successful decoding can be written in terms of \emph{collision relative entropy}~\cite{RennerThesis}. Then to find a lower bound on the expectation value of the probability of successful decoding (over the random choice of codewords) we use the joint convexity of the exponential of collision relative entropy and apply Jenssen's inequality. This, for instance, gives a lower bound on the one-shot capacity of c-q channels in terms of collision relative entropy. To compute the asymptotics of this  bound we prove a lower bound on the collision relative entropy in terms of another quantity called \emph{information spectrum relative entropy} whose asymptotics, up to the second order, has been computed by Tomamichel and Hayashi~\cite{TomamichelHayashi}. This gives a simple proof for the achievability parts of the second order asymptotics of the problems of the capacity of c-q channels recently computed by Tomamichel and Tan~\cite{MarcoTan}, quantum hypothesis testing previously computed by Tomamichel and Hayashi~\cite{TomamichelHayashi} and Li~\cite{Li}, and classical data compression with quantum side information previously derived by Tomamichel and Hayashi~\cite{TomamichelHayashi}.  \\

\noindent \textbf{Related works:}
One-shot achievable bounds (and converses) for c-q channel coding were also derived by Wang and Renner \cite{WangRenner} and Datta et al.~\cite{DMSB} (see also \cite{HayashiNagaoka}). The one-shot achievability result of~\cite{WangRenner} has very recently been used by Tomamichel and Tan \cite{MarcoTan} to compute the dispersion of c-q channels. 
This one-shot achievable bound is proved using the Hayashi-Nagaoka operator inequality~\cite[Lemma 2]{HayashiNagaoka}. To prove our results however, we do not use this inequality. 
Quantum hypothesis testing and data compression with quantum side information in the one-shot case have also been studied in~\cite{DMSB,Hayashi}, and~\cite{RennerRenes} respectively.

\subsection{Notation}\label{sec:notation}
 Throughout this paper we assume that all random variables have finite alphabets and quantum systems have finite dimensions. 
Quantum and classical systems are denoted by uppercase letters $B, X$ etc. 
We save $X$ to denote classical systems whose alphabet set is denoted by calligraphic letter $\mathcal X$. So a density matrix $\rho_X$ over $X$ is determined by its diagonal elements that are indexed by elements $x\in \mathcal X$, i.e.,
$\rho_X=\sum_{x\in \mathcal X} p_X(x) \ket x\bra x,$
where $p_X$ is a distribution over $\mathcal X$. 
In this paper we consider samples from this distribution whose outcomes are elements of $\mathcal X$. However, to distinguish elements $x\in \mathcal X$ with samples drawn from $p_X$ we denote the latter by \emph{mathtt} lowercase $\mathtt x$. As a result, letting $\rho_{XB}$ to be a density matrix over the joint system $XB$, 
$$\rho_{XB} = \sum_{x\in \mathcal X} p_X(x)\ket x\bra x\otimes \rho_x,$$
it is important to differentiate between $\rho_x$ and $\rho_{\tx}$; while the former is the state of $B$ when the subsystem $X$ is in the state $x$, the latter is a random density matrix, taking the value of $\rho_x$ when $\tx=x$. Needless to say, $\rho_X$ is the marginal density matrix over $X$.

For a positive semidefinite $U$ we use $U^{-1}$ for the inverse of $U$ restricted to $\supp(U)$ which is the span of eigenvectors of $U$ with non-zero eigenvalues. Thus $UU^{-1}=U^{-1}U$ is equal to the projection on $\supp(U)$ and not necessarily $I$, the identity operator. 

All logarithms and exponential functions in the paper are in base $2$.

 
\section{Collision relative entropy and information spectrum}
 
For density matrix $\rho$ and positive semidefinite (not necessarily normalized) $\sigma$, the collision relative entropy or collision divergence is defined by  
\begin{align}\label{eq:D2}
D_2(\rho \| \sigma) = \log \tr\left[    \left(  \sigma^{-\frac 1 4} \rho \sigma^{-\frac 1 4}  \right)^2  \right],
\end{align} 
if $\supp(\rho)\subseteq \supp(\sigma)$, and by $D_2(\rho\| \sigma)=\infty$ otherwise. 
This quantity (as a conditional entropy) was first introduced in \cite[Definition 5.3.1]{RennerThesis} and has found applications in quantum information theory~\cite{MCW2, DFW2, DBWR}. Collision divergence is a member of the family of quantum (sandwiched) R\'enyi divergences defined in \cite{Lennertetal, Wildeetal} (and before that in talks \cite{TMarco} and \cite{Fehr}), and further studied in~\cite{FrankLieb, Beigi2}. Some of the known properties of $D_2(\cdot \| \cdot )$  are as follows:

\begin{theorem}{\rm \cite{Lennertetal, Wildeetal}}
Collision relative entropy satisfies the following properties:
\begin{itemize}
\item \emph{(Positivity)} For (normalized) density matrices $\rho, \sigma$ we have $D_2(\rho\| \sigma) \geq 0$.
\item \emph{(Data processing)} For any quantum channel $\mathcal N$  we have $D_2(\rho\| \sigma) \geq D_2(\mathcal N(\rho)\| \mathcal N(\sigma))$.
\item \emph{(Joint convexity)} $\exp D_2(\cdot \| \cdot)$ is jointly convex, i.e., for density matrices $\rho_x$ and positive semidefinite $\sigma_x$ and probability distribution $p(x)$ we have
$$\sum_x p(x) \exp D_2(\rho_x\| \sigma_x) \geq \exp D_2(\rho \| \sigma),$$
where $\rho = \sum_x p(x)\rho_x$ and $\sigma=\sum_x p(x)\sigma_x$. 
\item \emph{(Monotonicity in $\sigma$)} If $\sigma'\geq \sigma$ then
$D_2(\rho \|\sigma')\leq D_2(\rho\| \sigma).$
\end{itemize}
\label{thm:collision-pr}
\end{theorem}

The other important quantity that we use in this paper, is the \emph{information spectrum relative entropy} introduced by Tomamichel and Hayashi~\cite{TomamichelHayashi}. We need the following notation to define this relative entropy.

For a hermitian matrix $U$ let $\Pi_{U\geq 0}$ be the \emph{orthogonal} projection onto the union of eigenspaces of $U$ with non-negative eigenvalues. Also for two hermitian operators $U, V$ let $\Pi_{U\geq V} = \Pi_{U-V\geq 0}.$
Now for  density matrix $\rho$ and positive semidefinite $\sigma$  and $\epsilon>0$, define the information spectrum relative entropy by  
$$D_s^\epsilon(\rho\| \sigma) := \sup\big\{ R | \, \tr\big(\rho\Pi_{\rho\leq 2^R \sigma}\big)\leq\epsilon\big\}.$$

Information spectrum relative entropy is used in~\cite{TomamichelHayashi} to compute the second order asymptotics of some information processing tasks including quantum hypothesis testing and source coding with quantum side information. The reason for introducing $D_s^\epsilon(\cdot \| \cdot)$ is that its second order asymptotics can be computed in terms of relative entropy and relative entropy variance. To express this result we first need the following definition:

\begin{definition}{\rm{\cite{TomamichelHayashi,Li}}} \label{def:infvar} The  information variance or the relative entropy variance
 $V(\rho\|\sigma)$ is defined by
$$V(\rho\| \sigma) := \tr\big(\rho(\log \rho - \log \sigma - D(\rho\|\sigma))^2\big).$$
Here $D(\rho \| \sigma)$ is the  relative entropy $D(\rho \| \sigma) := \tr(\rho(\log \rho - \log \sigma)).$
\end{definition}

\begin{theorem} {\rm \cite{TomamichelHayashi}} \label{thm:Ds-asymptotics}
For every two density matrices $\rho, \sigma$, fixed $0< \epsilon<1$ and $\delta=O(1/\sqrt n)$ we have
$$D_s^{\epsilon\pm \delta} (\rho^{\otimes n}\| \sigma^{\otimes n}) = nD(\rho\| \sigma) + \sqrt{n V(\rho\| \sigma)} \Phi^{-1}(\epsilon) +  O(\log n),$$
where $\Phi$ is the cumulative distribution of the standard normal distribution
$\Phi(u) :=\int_{-\infty}^u \frac{1}{\sqrt{2\pi}} e^{-\frac{1}{2}t^2 }{\rm{d}} t,$ and $\Phi^{-1}(\epsilon)= \sup\{u|\, \Phi(u)\leq \epsilon\}$. 
\end{theorem}

 Since this result is not explicitly proved in~\cite{TomamichelHayashi}, here for the sake of completeness we present a proof.

\begin{proof} 
Equation (34) of~\cite{TomamichelHayashi} states that 
$$D_h^{\epsilon}(\rho^{\otimes n}\| \sigma^{\otimes n})=nD(\rho\| \sigma) + \sqrt{nV(\rho\| \sigma)}\Phi^{-1}(\epsilon) +O(\log n),$$
where $D_h^{\epsilon}(\cdot\| \cdot)$ denotes the \emph{hypothesis testing} relative entropy (see~\cite{TomamichelHayashi} for its definition).
This equation is proved in~\cite{TomamichelHayashi} using equation (33) and the inequality in the beginning of Section VI A of~\cite{TomamichelHayashi}. It is not hard to see that by the same argument for every $\epsilon>0$ and $\delta=O(1/\sqrt n)$ we have 
$$D_h^{\epsilon\pm \delta}(\rho^{\otimes n}\| \sigma^{\otimes n})=nD(\rho\| \sigma) + \sqrt{nV(\rho\| \sigma)}\Phi^{-1}(\epsilon) +O(\log n).$$
On the other hand, from Lemma 12 of~\cite{TomamichelHayashi} for $0< \delta'<\epsilon$ we have
$$D^{\epsilon-\delta'}_h(\rho^{\otimes n}\|\sigma^{\otimes n})+\log \delta' \leq 
D_s^{\epsilon}(\rho^{\otimes n}\| \sigma^{\otimes n}) 
\leq D_h^{\epsilon}(\rho^{\otimes n}\| \sigma^{\otimes n}).$$
The proof is then completed by putting the above two equations together
\end{proof}

In this paper we derive our achievability bounds in terms of collision relative entropy. In the following theorem we prove a lower bound on the collision relative entropy in terms of information spectrum relative entropy.  We will use this lower bound and Theorem~\ref{thm:Ds-asymptotics} to compute the asymptotics of our achievable bounds.

\begin{theorem} \label{thm:collision-spectrum}
For every $0< \epsilon, \lambda <1$, density matrix $\rho$ and positive semidefinite $\sigma$ we have
\begin{align*}\exp D_2&(\rho\| \lambda \rho  + (1-\lambda)\sigma) \geq
\\& (1-\epsilon)\left[\lambda + (1-\lambda) \exp\big(-D_s^{\epsilon}(\rho \| \sigma) \big)\right]^{-1}.\end{align*}
\end{theorem}

\begin{proof} Let $\Pi = \Pi_{\rho\leq 2^R \sigma}$ and $\Pi' = I-\Pi = \Pi_{\rho> 2^R \sigma}$ where $R$ is a real number to be determined. Define the pinching map
$\mathcal N (\rho) := \Pi \rho \Pi  + \Pi'\rho \Pi'.$
Then the following chain of inequalities hold.
\begin{align}
\exp &D_{2}(\rho\| \lambda \rho + (1-\lambda)\sigma)  \nonumber
\\&\geq \exp D_2(\mathcal N(\rho) \| \lambda \mathcal N(\rho) + (1-\lambda)\mathcal N(\sigma))\label{eqn:aaaab1} \\
& \geq \exp D_2(\Pi'\rho \Pi' \| \lambda \Pi'\rho \Pi' + (1-\lambda)\Pi'\sigma \Pi')\label{eqn:aaaab2}\\
& \geq \exp D_2(\Pi'\rho \Pi' \| \lambda \Pi'\rho \Pi' + (1-\lambda)2^{-R}\Pi'\rho \Pi')\label{eqn:aaaab3}\\
& = \left(  \lambda + (1-\lambda)2^{-R}  \right)^{-1} \tr(\Pi' \rho).\nonumber
\end{align}
Here~\eqref{eqn:aaaab1} comes from the data processing inequality for collision relative entropy. For~\eqref{eqn:aaaab2} observe that since $\Pi$ and $\Pi'$ are orthogonal to each other, $\exp D_2(\mathcal N(\rho)  \| \lambda \mathcal N(\rho) + (1-\lambda)\mathcal N(\sigma))$ is a summation of two non-negative terms. For~\eqref{eqn:aaaab3} we use the last property of collision relative entropy stated in Theorem~\ref{thm:collision-pr} and that by the definition of $\Pi'$ we have $\Pi'(\rho - 2^R\sigma)\Pi'\geq 0$. Now the result follows by letting $R=D_s^{\epsilon}(\rho\| \sigma)$ and using the fact that by the definition of information spectrum relative entropy we have $\tr(\Pi'\rho)\geq 1-\epsilon$.
\end{proof}

\section{A one-shot achievable bound for c-q channels}\label{sec:sectionII}
Consider a c-q channel with input $X$ and output $B$ which maps $x\in \mathcal X$ to $\rho_x$. Fix an arbitrary distribution
$p_{X}(x)$ on the input. This distribution induces a density matrix over the joint system $XB$:
\begin{align}\label{eq:rho-x-b-1}
\rho_{XB} = \sum_{x\in \mathcal X} p_X(x)\ket x\bra x\otimes \rho_x.
\end{align}
As usual, to communicate a classical message $m\in \{1, \dots, M\}$ over the channel, we encode the message $m$ with $\tx_m$. In other words, we generate the random codebook $\mc=\{\tx_1,\cdots,\tx_M\}$ where the elements $\tx_m$ are drawn independently from $p_{X}(x)$. So if the message to be communicated is $m$, the output of the channel is $\rho_{\tx_m}$. As mentioned above, Yassaee et al.~\cite{Yassaee} use SLC to decode the message. Putting in the terminology of quantum information theory and applying it to c-q channels, this decoder is nothing but a pretty good measurement. So we assume that to decode the message the receiver applies pretty good measurement corresponding to signal states $\{\rho_{\tx_1}, \dots, \rho_{\tx_M}\}$ with POVM elements
\begin{align}
E_m&=\Big(\sum_{i}\rho_{\tx_{i}}\Big)^{-\frac{1}{2}}\rho_{\tx_m}\Big(\sum_{i}\rho_{\tx_{i}}\Big)^{-\frac{1}{2}}
\label{eq:POVMeq}.
\end{align}
There is nothing new in quantum information theory up to this point. The error analysis of this encoder/decoder however contains the main idea of this work, and as in~\cite{Yassaee} is based on Jensen's inequality. In the following we find that the probability of successful decoding can be written in term of collision relative entropy and then we use joint convexity. 

\begin{theorem}\label{thm:oshot1}
The \emph{expected value} of the probability of successful decoding of the pretty good measurement~\eqref{eq:POVMeq} corresponding to a randomly generated codebook according to distribution $p_X$ is bounded by
\begin{align*}
\e_{\mc} \Pr[\text{\rm{succ}}]
&\geq M^{-1} \exp D_2\left(\rho_{XB} \big\|\, \frac{1}{M}  \rho_{XB} + \big(1-\frac{1}{M}\big) \rho_X\otimes \rho_B\right),
\end{align*}
where $\rho_{XB}$ is defined in~\eqref{eq:rho-x-b-1}.
\end{theorem} 

For a classical channel, the above bound reduces to Theorem 1 of \cite{Yassaee}, although it is not expressed in terms of collision relative entropy in \cite{Yassaee}.

\begin{remark}
One might guess that replacing $ \frac{1}{M}  \rho_{XB} + \big(1-\frac{1}{M}\big) \rho_X\otimes \rho_B$ by $ \rho_X\otimes \rho_B$ may also result in a valid lower bound. This is incorrect since this would imply that the rate  $D_2(\rho_{XB}\|\rho_X\otimes \rho_B)\geq D(\rho_{XB}\|\rho_X\otimes \rho_B)=I(X;B)_{\rho}$  is asymptotically achievable with zero probability of error, which we know is not the case. Indeed we have examples of channels with integer capacity, whose capacity cannot be achieved with success probability of one. The above inequality between collision relative entropy and relative entropy is proved in~\cite{Lennertetal, Wildeetal, Beigi2}.
\end{remark}

\begin{proof}
Let us define 
$$\sigma_{UXB} = \frac{1}{M}\sum_{m=1}^M \ket m\bra m\otimes \ket{\tx_m} \bra{\tx_m}\otimes \rho_{\tx_m},$$
and let $\sigma_{UX}$ and $\sigma_B$ be its corresponding marginal density matrices. Note that $\sigma_{UXB}$ is a random density matrix. A direct computation verifies that the probability of successful decoding is equal to
\begin{align*} \Pr[\text{\rm{succ}}]& =\frac 1 M\sum_{m=1}^M \tr\left(E_m \rho_{\tx_m} \right)\\
&= \frac{1}{M} \tr \left[ \sum_{m=1}^M \left(    \Big(\sum_i \rho_{\tx_i}\Big)^{-1/4} \rho_{\tx_m} \Big(\sum_i \rho_{\tx_i}\Big)^{-1/4}   \right)^2  \right]\\
& = \frac{1}{M} \exp D_2( \sigma_{UXB} \| \sigma_{UX} \otimes \sigma_{B}).
\end{align*}
Here the last step is our key observation; The success probability of pretty good measurement (transpose channel in general) can be written in terms of collision relative entropy.

Now using the data processing inequality and joint convexity properties of Theorem~\ref{thm:collision-pr},
the expected value of the probability of correct decoding with respect to randomly chosen codewords, is lower bounded using Jensen's inequality as follows:
\begin{align*}
\e_{\mc} \Pr[\text{\rm{succ}}]& = \frac{1}{M} \e_{\mc} \exp D_2( \sigma_{UXB} \| \sigma_{UX} \otimes \sigma_{B})\\
&\geq \frac{1}{M} \e_{\mc} \exp D_2(\sigma_{XB} \| \sigma_X\otimes \sigma_B)\\
& \geq \frac{1}{M} \exp D_2 \big(  \e_{\mc}\, \sigma_{XB} \| \e_{\mc}\, \sigma_X\otimes \sigma_B  \big).
\end{align*}
Computing $\e_{\mc}\, \sigma_{XB}$ and $\e_{\mc}\, \sigma_X\otimes \sigma_B$ gives the desired result:
\begin{align*}
\e_{\mc}\, \sigma_{XB}  &= \e_{\mc} \frac{1}{M}\sum_m  \ket{\tx_m}\bra{\tx_m}\otimes \rho_{\tx_m}  
\\&= \e_{\mc}\, \ket{\tx_1}\bra{\tx_1} \otimes \rho_{\tx_1}  \\&= \rho_{XB}.\\
\e_{\mc}\, \sigma_X\otimes \sigma_B & = \e_{\mc}\, \frac{1}{M^2} \sum_{m, m'} \ket{\tx_m}\bra{\tx_{m}} \otimes \rho_{\tx_{m'}}  \\
& =  \e_{\mc} \frac{1}{M^2}\sum_m \ket{\tx_m}\bra{\tx_m}\otimes \rho_{\tx_m}\\&\qquad +  \e_{\mc} \frac{1}{M^2} \sum_{m\neq m'}  \ket{\tx_m}\bra{\tx_{m}} \otimes \rho_{\tx_{m'}} \\
& = \frac{1}{M} \rho_{XB} + \big(1-\frac{1}{M}\big) \rho_X\otimes \rho_B.\qedhere 
\end{align*}
\end{proof}

Combining the above theorem with Theorem~\ref{thm:collision-spectrum} we obtain the following one-shot lower bound:
\begin{corollary}  \label{corol:os2} For every distribution $p_X$ and $\epsilon>\delta>0$, it is possible to transmit one out of $M$ messages using a single use of the c-q channel, with the average probability of error being at most $\epsilon$, provided that
$$M= \left\lfloor \frac{\epsilon - \delta}{1-\epsilon}\, \exp\big(D_s^{\delta}(\rho_{XB}\| \rho_{X}\otimes \rho_B)\big) +1\right\rfloor.$$

\end{corollary}

\subsection{Second order asymptotics of c-q channels}\label{sec:soa}

For a channel $x\mapsto \rho_x$ let $M^*(n, \epsilon)$ be the maximum size of a codebook with codewords of length $n$ whose average probability of successful decoding (under the optimal decoding algorithm) is at least $1-\epsilon$. Then by the HSW theorem $\log M^*(n, \epsilon)$ (for small $\epsilon>0$)  is roughly equal to $nC$ where $C$ is the Holevo information of the channel given by
\begin{align}\label{eq:capacity}
C=\max_{p_X} I(X; B)_{\rho}.
\end{align}
Here $I(X; B)_{\rho}$ denotes the mutual information corresponding to the joint state~\eqref{eq:rho-x-b-1}. Our goal in the second order analysis is to find a more accurate estimate of $\log M^*(n, \epsilon)$. 
Based on the method of types, it is already shown by Winter~\cite{Winter} that 
$\log M^*(n, \epsilon) = nC + O(\sqrt n).$
So we may write $\log M^*(n, \epsilon) = nC + \sqrt n f(\epsilon) + o(\sqrt n)$. By computing the second order asymptotics we mean finding $f(\epsilon)$.

The second order asymptotics of c-c channels was first computed by Strassen~\cite{Strassen}. Under a mild condition on the channel, he showed that for every $0<\epsilon<1/2$,
\begin{align}\label{eq:st-dispersion}
\log M^*(n, \epsilon) = nC + \sqrt{nV} \Phi^{-1}(\epsilon) + o(\log n).
\end{align}
Here $V$ is a parameter of the channel (independent of $\epsilon$) which following~\cite{PPV} we call \emph{channel dispersion}.  Channel dispersion has recently been studied further by Polyanskiy et al.~\cite{PPV} and Hayashi~\cite{Hayashi2}. For c-q channels, the second order asymptotics has been computed very recently by Tomamichel and Tan~\cite{MarcoTan}. Here we re-derive the achievability part of their result. 


\begin{theorem}\label{thm:channeldispersion}
For every c-q channel $x\mapsto \rho_x$ and $0<\epsilon<1$ we have
$$\log M^*(n, \epsilon) \geq nC + \sqrt{nV_{\epsilon}} \Phi^{-1}(\epsilon) + O(\log n),$$
where $C$ is the channel capacity given by~\eqref{eq:capacity}. Also, $V_\epsilon$ is the channel dispersion given by 
$$V_{\epsilon}=\begin{cases}
\min_{p_{X}\in \mathcal P}V(\rho_{XB}\|\rho_X\otimes \rho_B), \qquad 0<\epsilon<1/2,\\
\max_{p_{X}\in \mathcal P}V(\rho_{XB}\|\rho_X\otimes \rho_B), \qquad 1/2\leq\epsilon<1,
\end{cases}
$$
where $\rho_{XB}$ is defined in~\eqref{eq:rho-x-b-1}, $\mathcal P$ is the set of capacity achieving input distributions, i.e. the set of distributions $p_X$ that achieve the optimal value in~\eqref{eq:capacity}, and $V(\cdot\|\cdot)$ is given in Definition \ref{def:infvar}.
\end{theorem}

\begin{proof} By Corollary~\ref{corol:os2}, for every distribution $p_X$ and $\epsilon>\delta>0$ we have
$$\log M^*(n,\epsilon) \geq D_s^{\delta} (\rho_{XB}^{\otimes n} \| \rho_X^{\otimes n}\otimes \rho_B^{\otimes n}) + \log\left( \frac{\epsilon -\delta}{1-\epsilon}\right).$$
Letting $\delta = \epsilon - 1/\sqrt n$, using Theorem~\ref{thm:Ds-asymptotics}, and optimizing over $p_X$ give the desired result.
\end{proof}

In particular, this theorem finds a new proof for the achievability part of the HSW theorem (i.e., maximum mutual information is an achievable rate for c-q channel coding) via a pretty good measurement that is \emph{directly} applied to signal states. To the best  knowledge of the authors, none of the previous proofs of the achievability part of HSW has this feature.

\section{Quantum hypothesis testing}
Suppose that a physical system is randomly prepared in one of the two states $\rho, \sigma$, which are called the hypotheses. To distinguish which hypothesis is the case we apply a POVM measurement $\{F_{\rho}, F_{\sigma}\}$ on the system. Such a measurement may cause an error in detecting the right hypothesis, and the goal of the hypothesis testing problem is to find the smallest possible probability of such an error.  Indeed there are two types of error:
Type I error is defined to be the probability of mis-detecting the hypothesis when the system is prepared in state $\rho$, and Type II error is defined similarly when the system is prepared in state $\sigma$. Then we have 
$$\Pr[\text{type I error}] = \tr(\rho F_{\sigma}), $$ and $$\Pr[\text{type II error}]=\tr(\sigma F_{\rho}).$$

In the asymmetric hypothesis testing problem we assume that the cost associated to type II error is much higher than the cost corresponding to type I error. So the probability of type II error should be minimized, while we only put a bound on the probability of type I error.   Quantum Stein's lemma~\cite{HiaiPetz, ON00} states that for every $\epsilon>0$, there is a POVM $\{F_{\rho}^{(n)}, F_{\sigma}^{(n)}\}$ 
to distinguish the hypotheses $\rho^{\otimes n}$ and $\sigma^{\otimes n}$  
 such that 
$\Pr[\text{type I error}]\leq \epsilon$ and 
$$\Pr[\text{type II error}]\leq \exp\big(-nD(\rho\| \sigma) + o(n)\big).$$ Moreover, $D(\rho\| \sigma)$ is the optimal such error exponent.
Also the one-shot hypothesis testing problem has been studied in~\cite{DMSB, Hayashi}.
Here we first prove a one-shot bound for the quantum hypothesis problem and then compute the asymptotics of our bound.

\begin{theorem}\label{thm:thm4} For every $\epsilon> \delta>0$ there is a POVM measurement $\{F_\rho, F_\sigma\}$ for the one-shot hypothesis testing problem such that 
$\Pr[\text{\rm{type I error}}]  \leq \epsilon$ and 
$$\Pr[\text{\rm{type II error}}]  \leq \exp\big(-D_s^{\epsilon-\delta}(\rho\| \sigma) - \log \delta\big).$$
\end{theorem}
\begin{proof}
Consider the POVM with elements 
$$
F_{\rho}=(\rho + M\sigma)^{-1/2} \rho (\rho+M\sigma)^{-1/2}$$and $$\qquad F_{\sigma}=(M^{-1}\rho + \sigma)^{-1/2} \sigma (M^{-1}\rho+\sigma)^{-1/2},
$$
where $M$ is a positive real number to be determined. Observe that the choice of this POVM is motivated by the proof of Theorem~\ref{thm:oshot1}. The probability of type I error is equal to 
$$\Pr[\text{type I error}]  = 1-\tr\left( \rho F_{\rho} \right) =1-\exp D_2(\rho\| \rho+ M\sigma).$$
Then using Theorem~\ref{thm:collision-spectrum} we have
\begin{align*}
1-\Pr[&\text{type I error}] \\
&\geq (1-(\epsilon-\delta))\left(  1+ M\exp\big( -D_s^{\epsilon-\delta}(\rho\| \sigma)   \big)   \right)^{-1}.
\end{align*}
The probability that type II error event does not occur, is equal to 
\begin{align*}
1- &\Pr[\text{type II error}] 
\\& =\exp D_2(\sigma \| M^{-1}\rho+ \sigma)\\
&= (1+M^{-1})^{-1}  \exp D_2\Big(\sigma  \big\| \frac{M^{-1}}{1+ M^{-1}}\rho + \frac{1}{1+M^{-1}}\sigma\Big)\\
& \geq (1+M^{-1})^{-1}, 
\end{align*}
where in the last line we use the positivity of collision relative entropy (see Theorem~\ref{thm:collision-pr}). Letting $M=\exp\big(D_s^{\epsilon-\delta}(\rho\| \sigma) + \log \delta \big)$ gives the desired result.
\end{proof}

The measurement $\{F_\rho, F_\sigma\}$ used in the above proof is a pretty good measurement associated to unnormalized signal states $\{\rho, M\sigma\}$. Should we allow general types of measurement, we can let $F_\rho=\Pi_{\rho>2^R\sigma}$ and $F_\sigma=\Pi_{\rho\leq 2^R\sigma}$ for an appropriate choice of $R$. With this measurement we get an even better bound, comparing to that given in Theorem~\ref{thm:thm4}. However, the significance of the above proof is that a pretty good measurement gives a one-shot bound for the hypothesis testing problem.

The second order asymptotics of quantum hypothesis testing has been found independently in \cite{TomamichelHayashi} and \cite{Li}. The achievability part of their result can be derived using Theorem~\ref{thm:Ds-asymptotics} and Theorem~\ref{thm:thm4} for $\delta=1/\sqrt n$.

\begin{theorem} 
For every $\epsilon>0$ and $n$ there exists a POVM  to distinguish $\rho^{\otimes n}$ and $\sigma^{\otimes n}$ such that
$\Pr[\text{\rm{type I error}}] \leq \epsilon$ and that
\begin{align*}-\log \Pr&[\text{\rm{type II error}}]\geq nD(\rho\| \sigma) + \sqrt{n V(\rho\| \sigma)}\Phi^{-1}(\epsilon) + O(\log n). 
\end{align*}
\end{theorem}

\section{Data compression with quantum side information}
Let 
$$\rho_{XB}= \sum_x p_X(x)\ket x\bra x\otimes \rho_x,$$
be a c-q state. Suppose that Alice receives $x$ with probability $p_X(x)$, in which case Bob receives system $B$ in state $\rho_x$. The goal of Alice is to transmit $x$ to Bob by sending a message $m$. To this end, they may fix an encoding hash function $g:\mathcal X\rightarrow \{1, \dots, M\}$. Then Alice after receiving $x$ sends $m=g(x)$ to Bob.  Now, Bob has access to $\rho_x$ and $m=g(x)$; his goal is to guess $x$. He knows that $x$ is in the set $g^{-1}(m)=\{x'|\, g(x')=m\}$. So to pick an element of $g^{-1}(m)$, Bob applies some measurement 
$$\{F_{x'}^{m}|\, x'\in \mathcal X, g(x')=m\},$$ 
on his state $\rho_x$. Note that this measurement depends on Alice's message $m$ and its outcome is some element of $g^{-1}(m)$. Assuming that Alice's input is $x$, the probability that Bob successfully decodes $x$ is then equal to $\tr(\rho_xF_x^{m})$ where $m=g(x)$. This can  be equivalently written as 
$$\Pr[\text{succ} | X=x] = \sum_m \mathbf{1}_{\{m=g(x)\}}  \tr(\rho_xF^m_x),$$
where $\mathbf{1}_{\{m=g(x)\}}$ is equal to $1$ if $g(x)=m$ and is equal to $0$ otherwise.
As a result, 

the probability of successful decoding is equal to 
\[\Pr[\text{\rm{succ}}] = \sum_{x, m} p_X(x)\mathbf{1}_{\{m=g(x)\}}  \tr(\rho_xF^m_x).\]

 This problem is called classical source coding (or data compression) with quantum side information, also known as the c-q Slepian-Wolf problem. Needless to say, the goal is to minimize $M$ while keeping the average probability of successful decoding close to one.  It is shown in~\cite{cqSW} that the optimal asymptotic rate of communication required to achieve this goal is the conditional entropy $H(X|B)_{\rho}$. A one-shot achievable and converse bound for this problem is derived in~\cite{RennerRenes}.
Moreover, the second order asymptotics of this problem has been computed in~\cite{TomamichelHayashi}. Here we prove a one-shot achievable bound for this problem.

\begin{theorem} For every $\epsilon>0$, there is an encoding map $g:\mathcal X\rightarrow \{1, \dots, M\}$ and decoding measurements   $\{F_x^{m}|\, x\in \mathcal X, g(x)=m\}$ such that the probability of correct decoding is bounded by 
\begin{align*}\Pr[\text{\rm{succ}}] 
&\geq 
\exp D_2\left(\rho_{XB} \big\|   \big(1-\frac{1}{M}\big) \rho_{XB} + \frac{1}{M} I_X\otimes \rho_B   \right)\\&\geq
\frac{1-\epsilon}{1+ M^{-1}\exp\big(-D_s^{\epsilon}(\rho_{XB}\| I_X\otimes \rho_B)\big)}.\end{align*}
\label{thm:thm6}
\end{theorem}
\begin{proof}
Suppose that Alice's encoding map $g:\mathcal X\rightarrow \{1, \dots, M\}$ is  chosen randomly, i.e., $g(x)$ is chosen uniformly at random and independent of other $g(x')$, $x'\neq x$. As before we use $\th(x)$ (in \emph{mathtt} format) to denote this random function.
Let the decoding POVM of Bob when receiving $m\in \{1, \dots, M\}$ be $\{F_x^{m}|\, x\in\mathcal X,  \th(x)=m\}$ where $F_x^m$ is defined by
 \begin{align}
F_x^m= p_X(x)\Big(\sum_{x'}   p_X(x')\rho_{x'} \mathbf{1}_{\{\th(x')=m\}}   \Big)^{-1/2} \rho_x \Big(\sum_{x'}   p_X(x')\rho_{x'} \mathbf{1}_{\{\th(x')=m\}}   \Big)^{-1/2}.\label{eqn:POVMBob}
\end{align}
  Observe that this POVM is indeed a pretty good measurement corresponding to the unnormalized signal states $\{p_X(x')\rho_{x'}|\, x'\in \th^{-1}(m)\}$. 
Then the expectation value of the average probability of correct decoding is equal to 
\begin{align*}
\E_{\th} \Pr[\text{\rm{succ}}]& = \E_{\th} \sum_{x, m} p_X(x) \mathbf{1}_{\{\th(x)=m\}}  \tr\left(  \rho_x   F_x^m    \right)
\\&= M \E_{\th} \sum_{x} p_X(x) \mathbf{1}_{\{\th(x)=1\}}  \tr\left(  \rho_x   F_x^1    \right).
\end{align*}
This can be written in terms of collision relative entropy as $\E_{\th} \Pr[\text{\rm{succ}}] =M\E_{\th} \exp D_2(\sigma_{XB} \| \tau_{XB} )$
where 
$$\sigma_{XB} = \sum_x p(x) \mathbf{1}_{\{\th(x)=1\}} \ket x \bra x\otimes \rho_x,$$
and $$
\tau_{XB}= \Big(\sum_{x} \mathbf{1}_{\{\th(x)=1\}} \ket x \bra x\Big)\otimes \Big(   \sum_{x'} p(x') \mathbf{1}_{\{\th(x')=1\}} \rho_{x'}   \Big).$$
Then using Jensen's inequality and the joint convexity of $\exp D_2(\cdot \| \cdot)$ we arrive at 
$$\E_{\th} \Pr[\text{\rm{succ}}] \geq M \exp D_2(\E_{\th}\, \sigma_{XB} \| \E_{\th}\, \tau_{XB}).$$
We have $\E_{\th}\, \sigma_{XB} = \frac{1}{M}\rho_{XB}$, and
\begin{align*}
\E_{\th} \tau_{XB} & = \E_{\th} \sum_{x, x'} p(x') \mathbf{1}_{\{\th(x)=1\}}\mathbf{1}_{\{\th(x')=1\}} \ket x\bra x\otimes \rho_{x'}\\ 
&=  \sum_{x, x'} p(x') \Big(\frac{1}{M} \mathbf{1}_{\{x=x'\}} + \frac{1}{M^2}\mathbf{1}_{\{x\neq x'\}}\Big) \ket x\bra x\otimes \rho_{x'}\\
& = \big( \frac{1}{M} - \frac{1}{M^2} \big)\rho_{XB} + \frac{1}{M^2} I_X\otimes \rho_{B}.
\end{align*}
Finally, using Theorem~\ref{thm:collision-spectrum} we obtain
\begin{align*}
\E_{\th}\Pr[\text{\rm{succ}}]  &\geq M \exp D_2 \left( \frac{1}{M}\rho_{XB} \big\|  \big( \frac{1}{M} - \frac{1}{M^2} \big)\rho_{XB} + \frac{1}{M^2} I_X\otimes \rho_{B}   \right) \\
&=\exp D_2\left(\rho_{XB} \big\|   \big(1-\frac{1}{M}\big) \rho_{XB} + \frac{1}{M} I_X\otimes \rho_B   \right)\\
& \geq (1-\epsilon)\left[  1-M^{-1} + M^{-1}\exp\big(- D_s^{\epsilon}(\rho_{XB}\| I_X\otimes \rho_B)\big)   \right]^{-1}\\
& \geq (1-\epsilon)\left[  1 + M^{-1}\exp\big(- D_s^{\epsilon}(\rho_{XB}\| I_X\otimes \rho_B)\big)   \right]^{-1}.\qedhere
\end{align*}
\end{proof}

\begin{corollary}
For every $\epsilon >\delta >0$ there is a protocol for classical data compression with quantum side information with
$$ M= \lceil \exp\big(    -D_s^{\delta}(\rho_{XB} \| I_X\otimes \rho_B) - \log(\epsilon -\delta)   \big) \rceil,$$
whose probability of error is at most $\epsilon$.
\end{corollary}

\begin{remark} 
Again using Theorem~\ref{thm:Ds-asymptotics} the achievability part of the asymptotic analysis of~\cite{TomamichelHayashi} for the c-q Slepian-Wolf problem can be derived from the above corollary.
\end{remark}


\section{Conclusion}
We proposed a partial quantum extension of the technique of \cite{Yassaee} and applied it to three basic information theoretic problems in the quantum case. In our generalization we noticed that some of the expressions of~\cite{Yassaee} can be written in terms of collision relative entropy, and used the joint convexity of its exponential. A full generalization of the technique of \cite{Yassaee} to more complicated scenarios such as channels with state, and quantum Marton coding requires new tools to meet the challenges of dealing with non-commuting operators as well as proving  (joint) operator convexity of certain functions to apply Jensen's inequality. 

We found one-shot achievability bounds for the problems of c-q channel coding, quantum hypothesis testing and source compression with quantum side information. From the expressions of these one-shot bounds, it is not clear to us how these bounds compare with previously existing bounds.


\section*{Acknowledgements}

The authors are thankful to unknown referees whose comments significantly improved the presentation of the paper.

\end{document}